\documentclass{rspublic}

\usepackage{pstricks}
\usepackage{calrsfs}

\usepackage{pifont}
\usepackage{oldgerm}
\usepackage{color} 
\usepackage{amssymb,amsbsy}
\usepackage{epic,eepic}
\usepackage{texdraw}
\usepackage{dsfont}
\usepackage{latexsym}

\usepackage{accents}

\DeclareMathAccent{\wtilde}{\mathord}{largesymbols}{"65}
\DeclareMathAccent{\what}{\mathord}{largesymbols}{"62}

\def\m@th{\mathsurround=0pt}
\mathchardef\bracell="0365 
\def\upbrall{$\m@th\bracell$}
\def\undertilde#1{\mathop{\vtop{\ialign{##\crcr
    $\hfil\displaystyle{#1}\hfil$\crcr
     \noalign
     {\kern1.5pt\nointerlineskip}
     \upbrall\crcr\noalign{\kern1pt
   }}}}\limits}

\newcommand{\wb}[1]{\overline{#1}}

\newtheorem{Def}{Definition}

\newtheorem{Pro}{Proposition}
\newtheorem{Lem}{Lemma}

\makeatletter
\def\wb{\accentset{{\cc@style\underline{\mskip10mu}}}}
\makeatother

\begin{document}

\title[On the Lagrangian formulation of multidimensionally consistent systems]{On the Lagrangian formulation of multidimensionally consistent systems}
\author[P. Xenitidis, F. Nijhoff \& S. Lobb]{Pavlos Xenitidis, Frank Nijhoff \& Sarah Lobb}
\affiliation{Department of Applied Mathematics, University of Leeds, Leeds LS2 9JT, UK}
\label{firstpage}
\maketitle

\begin{abstract}{Integrable systems, multidimensional consistency, Lagrangians, symmetries}
Multidimensional consistency has emerged as a key integrability property for partial difference equations (P$\Delta$Es) 
defined on the ``space-time'' lattice. It has led, among other major insights, to a classification of scalar affine-linear quadrilateral 
P$\Delta$Es possessing this property, leading to the so-called ABS list. Recently, a new variational principle has been 
proposed that describes the multidimensional consistency in terms of discrete Lagrangian multi-forms.  This description 
is based on a fundamental and highly nontrivial property of Lagrangians for those integrable lattice equations, namely the 
fact that on the solutions of the corresponding P$\Delta$E the Lagrange forms are closed, i.e. they obey a {\it closure relation}.  Here we extend those results to the continuous case: it is known that associated with the integrable P$\Delta$Es there exist 
systems of PDEs, in fact differential equations with regard to the parameters of the lattice as independent variables, 
which equally possess the property of multidimensional consistency. In this paper we establish a universal Lagrange structure 
for affine-linear quad-lattices alongside a universal Lagrange multi-form structure for the corresponding continuous PDEs, and 
we show that the Lagrange forms possess the closure property. 
\end{abstract}

\section{Introduction}

The study of integrable systems on the space-time lattice has, over the past decades, developed into a major area of research. The early 
examples of integrable lattice systems, i.e. partial difference equations (P$\Delta$Es) on two- or higher-dimensional grids, go back to the mid 
1970s and early 1980s, when the research was focused on discretizing known continuous soliton systems (Ablowitz \& Ladik 1976, Hirota 1977, Date {\it et al.} 1982,1983, Nijhoff {\it et al.} 1983, Quispel {\it et al.} 1984). However, in recent years the focus has shifted to studying these equations 
on their own merit, since they exhibit the key integrability aspects in a more lucid way than their continuous counterparts, and often 
can be thought of as more generic systems than the corresponding differential equations which can be retrieved by appropriate continuum limits\footnote{It should be pointed 
out that only very special and non-trivial discretizations of a given integrable differential equation will retain the key integrability properties of the latter.}. 

It has emerged that one of the key integrability aspects of lattice equations is the property of multidimensional consistency, i.e., the property that one can impose
P$\Delta$Es of the same form (differing only by a choice of the so-called lattice parameters) in different pairs of perpendicular directions of a multidimensional lattice in which 
the two-dimensional grids are embedded (Nijhoff \& Walker 2001, Bobenko \& Suris 2002). This phenomenon can be understood as the analogue of the existence of infinite 
hierarchies of PDEs associated with an integrable nonlinear evolution equation, such as the KdV equation. In some sense multidimensional consistency implies that the 
``integrable system'' is not really a single P$\Delta$E but comprises all members of the parameter-family of equations compatible on the multidimensional lattice 
(Nijhoff {\it et al.} 2001). Using this property a classification of two-dimensional scalar integrable lattice systems was given, in the affine-linear case, by Adler, 
Bobenko and Suris (2003) resulting in what is hereafter referred to as the ABS list. Although we shall refer to the equations in the list as the ``ABS equations'', which are given in the Appendix, 
several examples in the list, notably lattice systems of KdV type (associated with discretizations of the KdV and related PDEs), were already known in the literature, 
cf. Nijhoff \& Capel (1995) and references therein. The most general equation of the ABS list was found by Adler (1998) and it contains lattice parameters which are 
points on an elliptic curve, cf. also Adler \& Suris (2004).   

Considering the integrable system as a system of many equations imposed simultaneously on one and the same (possibly scalar) dependent variable, raises the question 
whether there exists a canonical description which encodes the whole system of equations. Even though most continuous integrable equations can be cast into Lagrangian form (Zakharov \& Mikhailov 1980), the situation for discrete equations is distinctly more complicated and, in fact, the first Lagrangian for a quadrilateral equation was given in Capel {\it{et al.}} (1991). Furthermore,
the usual least-action principle only provides one variational equation per component of the dependent variable. Lobb \& Nijhoff (2009) provided the first proposal 
on how to achieve a variational description of the multidimensionally consistent system of lattice equations, based on the idea of Lagrangian multi-forms. In fact, 
the proposal is to consider action functionals in terms of discrete Lagrange 2-forms $\mathcal{L}$ given by
 \begin{equation}\label{S} 
  S = S[u(\boldsymbol{n});\sigma] = \sum_\sigma{\cal{L}} = \sum_{\sigma_{ij}(\boldsymbol{n})\in\sigma}{{\cal{L}}_{ij}(\boldsymbol{n})},
 \end{equation} 
 where $u(\boldsymbol{n})$ is a field configuration of a dependent variable $u$ depending on lattice sites labelled by $\boldsymbol{n}\in\mathbb{Z}^N$ of a lattice of 
 arbitrary dimension $N>2$ and $\sigma$ denotes a quadrilateral surface, composed of a connected configuration of quadrilaterals, embedded in the $N$-dimensional grid. 
 The least-action principle formulated in Lobb \& Nijhoff (2009) imposes that $S$ attains a critical point not only with regard to the infinitesimal variations of 
 the field variables $u(\boldsymbol{n})$ for a given chosen surface $\sigma$ , but also with regard to the geometry of the independent variables, by imposing 
 surface independence of the action for the solutions of the Euler-Lagrange (EL) equations. The latter requirement is equivalent to stating that the Lagrangians 
 $\mathcal{L}_{ij}$, viewed as 2-forms associated with the quadrilaterals that compose the surface, have to be closed forms subject to the EL equations. It is this property that makes the choice of Lagrangians {\em admissable} and which immediately implies multidimensional consistency of the system of equations arising 
 from the EL equations on each choice of surface $\sigma$ (fixing, if necessary, boundaries). Remarkably, this closure property could be established for well-chosen 
 Lagrangians for equations in the ABS list, either by direct computation (as was done in Lobb \& Nijhoff 2009) or by general principles (Bobenko \& Suris 2010). 
 Furthermore, Lagrangians possessing the closure property were found in the case of multi-component systems (Lobb \& Nijhoff 2010) and higher-dimensional systems of 
 KP type (Lobb {\it et al.} 2009). 
 
The phenomenon of multidimensional consistency is not restricted to the discrete case of lattice equations. In Nijhoff {\it et al.} (2000) the 
novel class of non-autonomous PDEs was introduced which are the equations of the same dependent variable appearing in the lattice equations but viewed now as a function 
of a pair of (continuous) lattice parameters, denoted by $\alpha_i,\alpha_j$. The so-called {\emph{generating PDE}} of the KdV hierarchy, i.e. a PDE which generates all the
equations in the KdV hierarchy by systematic expansion with respect to the lattice parameters, arises as the EL equations ~$\delta_u\mathcal{L}=0$~ from a simple 
Lagrangian given by   
  \begin{equation}
  {\cal{L}} = \frac{1}{2}(\alpha_i-\alpha_j)\frac{\left(\partial_{\alpha_{i}} \partial_{\alpha_{j}} u\right)^2}{\partial_{\alpha_i} u\, \partial_{\alpha_j} u} + \frac{1}{2(\alpha_{i}-\alpha_{j})}\biggl(n_{j}^2
  \frac{\partial_{\alpha_{i}}u}{\partial_{\alpha_{j}} u}+n_{i}^2\frac{\partial_{\alpha_{j}} u}{\partial_{\alpha_{i}} u}\biggr). \label{eq:scal-Lag}
 \end{equation} 
In Lobb \& Nijhoff (2009) it was also shown that this continuous Lagrangian can be viewed as a closed 2-form, on solutions of the 
 EL equations, subject to a three-dimensional constraint on the solutions\footnote{In fact, this constraint is equivalent to stating that the solutions 
 obey a KP type of equation for the embedding in a higher dimension, which holds true for known infinite families of 
 solutions such as soliton solutions.}. Thus, a similar multidimensional least-action principle in terms of a continuous action functional 
  \begin{equation}
  S[u(\boldsymbol{\alpha});\sigma] \,= \,\int_\sigma\!\!{\mathcal{L}} \,=\,\int_{\sigma} \sum_{ i<j} {{\cal{L}}_{ij}\, \rd\alpha_{i}\wedge \rd\alpha_{j}}
 \end{equation} 
which involves continuous Lagrange 2-forms $\mathcal{L}$ depending on the fields $u(\boldsymbol{\alpha})$ and on surfaces $\sigma$ with coordinates given by the vector ${\boldsymbol{\alpha}}=(\alpha_1,\ldots,\alpha_N)$, where $i,j=1,\dots,N$, can be formulated in the continuous case for certain PDEs associated with KdV type lattice equations. In order to avoid the difficulty stemming from the fact 
 that the scalar Lagrangian \eqref{eq:scal-Lag}, and consequently also the EL equation, is of higher order it is natural to formulate the Lagrange structure to 
 that for a system of lower-order equations, which actually emerge naturally from the corresponding Lax pair. In fact, such systems of PDEs were subsequently 
 established for all equations in the ABS list through symmetry analysis (Tsoubelis \& Xenitidis 2009) and as a specific  
 example of the construction we can mention the following coupled system of PDEs, corresponding to H1, H2 and Q1: 
\begin{eqnarray}\label{H1etc-sys}
\frac{\partial u_i}{\partial \alpha_j} &= & \frac{u_i\,-\,u_j}{\alpha_i\,-\,\alpha_j}\,\left(n_j\,-\,(u_i\,-\,u_j)\,\frac{\partial u}{\partial \alpha_j}\right)\,+\, \delta^2(\alpha_i-\alpha_j)\,\frac{\partial u}{\partial \alpha_j}\,, \nonumber \\
\frac{\partial^2 u}{\partial \alpha_i \partial \alpha_j} &=& 2\,\frac{u_i\,-\,u_j}{\alpha_i\,-\,\alpha_j} \frac{\partial u}{\partial \alpha_i} \frac{\partial u}{\partial \alpha_j}\,+\,\frac{n_i}{\alpha_i\,-\,\alpha_j} \frac{\partial u}{\partial \alpha_j}\,+\,\frac{n_j}{\alpha_j\,-\,\alpha_i} \frac{\partial u}{\partial \alpha_i}\,.
\end{eqnarray}
It was derived in Tsoubelis \& Xenitidis (2009) along with its Lax pair, where $u$, $u_i$ and $u_j$ denote a triple of dependent variables. It is worth mentioning that system (\ref{H1etc-sys}) with $\delta =0$ was first given in Nijhoff {\it{et al.}} (2000) in relation to H1 and was derived from a reduction of the anti-self dual Yang-Mills equations in Tongas {\it{et al.}} (2001). In the present paper we will derive the Lagrange structure for the general class of systems of the form 
~${\cal L}(u,u_i,u_j;\partial_{\alpha_i}u,\partial_{\alpha_j}u,\partial_{\alpha_i}u_j,\partial_{\alpha_j}u_i)$~ such that the EL equations, when varying 
with respect to all three dependent variables, are satisfied on solutions of the systems like \eqref{H1etc-sys}. In this case the Lagrangian reads   
\begin{eqnarray}
{\cal{L}}_{ij} &=& \frac{\partial u}{\partial \alpha_i} \frac{\partial u_i}{\partial \alpha_j}- 
\frac{\partial u}{\partial \alpha_j} \frac{\partial u_j}{\partial \alpha_i}\,-\,\left(\delta^2 (\alpha_i-\alpha_j)\,
-\,\frac{(u_i-u_j)^2}{\alpha_i-\alpha_j}\right)\frac{\partial u}{\partial \alpha_i} \frac{\partial u}{\partial \alpha_j}\nonumber \\
& & -\,n_j\,\frac{u_i-u_j}{\alpha_i-\alpha_j}\,\frac{\partial u}{\partial \alpha_i} \,+ \,n_i\,\frac{u_j-u_i}{a_j-a_i} \,\frac{\partial u}{\partial \alpha_j} \label{H1etc-sys-Lagr}
 \end{eqnarray}
and one can show directly, without assuming further constraints on the higher-dimensional embedding, that this Lagrangian obeys the continuous version of the closure 
relation 
\begin{equation}
\partial_{\alpha_k}{\cal L}_{ij}+\partial_{\alpha_i}{\cal L}_{jk}+\partial_{\alpha_j}{\cal L}_{ki}=0 \,. 
\end{equation}
The organization of the paper is as follows. We first give in Section 2 a universal description of the Lagrange structure of affine-linear quadrilateral equations. 
What is important here is that the Lagrangians can actually be systematically derived from some basic assumptions, such as a 3-point form of the Lagrange function and 
the Kleinian symmetry of the quadrilateral equation. This leads to an integral form of the so-called three-leg representation of the quadrilateral equation, cf. 
Adler {\it et al.} (2003). In Section 3, we specialize to the ABS case, in which we have the symmetries of the square, imposing dependence on lattice 
parameters in a covariant way. We prove that the corresponding Lagrangians, in the integral form, obey the closure property. In Section 4, we briefly present the construction of the universal system of PDEs associated with the ABS list, which was given in Tsoubelis \& Xenitidis (2009), and we prove the multidimensional consistency of these systems. 
Finally in Section 5 we establish the corresponding Lagrangian structures, both for the systems of PDEs as well as for the differential-difference equations which 
are defined in terms of the symmetry generators, and show that the universal Lagrange structure for the system of PDEs obey the closure relation. 
Thus, we believe this completes the full picture of the Lagrange structure for multidimensionally consistent equations associated with the affine-linear scalar quad-lattices.

\section{Lagrangian formulation of 2D quad-lattices} \label{sec-Klein}

We first focus on the Lagrangian formulation of affine-linear quadrilateral equations on the two-dimensional lattice, and show that under very general assumptions there 
is an almost unique Lagrangian description. The members of the class of equations we consider will be denoted by 
\begin{equation}
Q(u,u_i,u_j,u_{ij})\,=\,0\,, \label{gen-eq}
\end{equation}
and involve the values of a function $u$ at four neighbouring points of an elementary quadrilateral. The labels $i$, $j$ refer, for the time being, to two 
fixed directions of a ${\mathds{Z}}^2$ lattice.  

The general notation we use here, adapted to cover also the situation we will encounter in the next section, is as follows. In principle, the unknown function 
$u$ of the equations will, in due course, be considered as functions of an 
arbitrary number of discrete variables, denoted by $n_i$, and there will be an equal number of continuous variables $\alpha_i$, i.e. 
$u = u(n_1,n_2,\ldots;\alpha_1,\alpha_2,\ldots)$ will be the dependent variable of a system of equations in an arbitrary number of 
dimensions. In that context, the shifted values of $u$ will be denoted by indices, e.g.
$$\begin{array}{l} u_i\,\equiv\,{\cal{S}}_i(u)\,:=\,u(\ldots,n_i+1,\ldots;\alpha),\\ u_{-j}\,\equiv\,{\cal{S}}_j^{-1}(u)\,:=\,u(\ldots,n_j-1,\ldots;\alpha),\\ 
u_{ij}\,\equiv\,{\cal{S}}_i{\cal{S}}_j(u)\,:=\,(\ldots,n_i+1,\ldots,n_j+1,\ldots;\alpha),\end{array}$$
where ${\cal{S}}_i$ is the shift operator in the $i$ direction. Occasionally we will also use the difference operator $\Delta_i$, which is defined by $\Delta_i\,:=\,{\cal{S}}_i\,-\,1$. 
The main integrability property of the systems considered in the following sections is multidimensional consistency. This means that all of them can be extended in 
a compatible way in higher dimensions, according to the following definition: 

\begin{Def} \label{MDC-dis}
Let $Q(u,u_i,u_j,u_{ij})=0$ denote a quadrilateral difference equation, to be solved for a dependent variable $u$. Imposing the dependence of $u$ on a third discrete variable $n_k$ such that
$$Q(u,u_j,u_k,u_{jk})=0\,,\quad Q(u,u_k,u_i,u_{ik})=0\,,$$
if the three different ways to evaluate the triply shifted value $u_{ijk}$ lead to the same result, then we call the equation defined by $Q$ multidimensionally consistent. 
\end{Def} 

However, in the present section we will only make statements about a single equation of the form \eqref{gen-eq}, in terms of two fixed independent variables $n_i$, $n_j$, 
and their corresponding lattice shifts. The following properties of the function $Q$ are assumed here: 
\begin{enumerate}
\item[1)]
$Q$ is an affine-linear polynomial depending explicitly on all the arguments, that is 
$\partial_u Q \,\partial_{u_i} Q \,\partial_{u_j} Q \,\partial_{u_{ij}} Q \,\ne\,0$ and
$\partial_u^2 Q \,=\,\partial_{u_i}^2 Q\,=\,\partial_{u_j}^2 Q\,=\,\partial_{u_{ij}}^2 Q\,=\,0$. 
\item[2)] 
$Q$ is irreducible, meaning that $Q$ cannot be factorized and presented as a product of two polynomials; 
\item[3)] 
$Q$ possesses the Kleinian symmetry, i.e., 
\begin{equation}
Q(u,u_i,u_j,u_{ij})\,=\,\epsilon\, Q(u_i,u,u_{ij},u_{j})\,=\,\sigma \, Q(u_j,u_{ij},u,u_i)\,, \label{gen-D4sym}
\end{equation}
where $\epsilon = \pm 1$ and $\sigma = \pm 1$.
\end{enumerate}

If $Q$ satisfies Property 1 one can define, in principle, six different bi-quadratic polynomials in terms of $Q$ and its derivatives, and four quartic polynomials 
through the application of the following double-sided Wronskian, respectively discriminant operator, defined by their actions on an arbitrary polynomial function 
$f = f(x,y)$ in the following ways, eg. Adler {\it et al.} (2009): 
\begin{equation}
{\cal{D}}_{x,y}(f)\,:=\,\left| \begin{array}{cc} f & \partial_x f \\ \partial_y f & \partial_x \partial_y f \end{array} \right|\,,  
\quad {\cal{D}}_x(f)=(\partial_x f)^2-2 f\partial_x^2 f\,.
\label{determ}
\end{equation}

Furthermore, the irreducibility of $Q$ guarantees that none of the polynomials obtained from $Q$ by acting on it by the operator ${\cal D}_{x,y}$ for all pairs of arguments 
$x,y$ of $Q$, is identically zero. The Kleinian symmetry (\ref{gen-D4sym}) implies that these bi-quadratic polynomials are symmetric and take the same form 
when we apply the operator ${\cal D}_{x,y}$ for mutually exclusive pairs of arguments. This is expressed by the following Lemma, which can be easily established by 
exploiting the definitions and symmetry relations: 
\begin{Lem} 
Due to the Kleinian symmetry \eqref{gen-D4sym}, the following polynomials assigned to an elementary quadrilateral are the same for functions depending 
on complementary pairs of variables:    
\begin{enumerate}
\item symmetric bi-quadratic polynomials assigned to the edges of the quadrilateral: 
\begin{eqnarray}
h(u,u_i) &:=& {\cal{D}}_{u_j,u_{ij}}(Q)\,,\quad h(u_j,u_{ij}) := {\cal{D}}_{u,u_i}(Q)\,, \nonumber \\
\bar{h}(u,u_j) &:=& {\cal{D}}_{u_i,u_{ij}}(Q) \,,\quad \bar{h}(u_i,u_{ij}) := {\cal{D}}_{u,u_j}(Q)\,, \label{poly-h-def}
\end{eqnarray}
 \item symmetric bi-quadratic polynomials assigned to the diagonals of the quadrilateral: 
\begin{equation} 
H(u_i,u_j) \,:=\, {\cal{D}}_{u,u_{ij}}(Q)\,,\quad
H(u,u_{ij}) \,:=\, {\cal{D}}_{u_i,u_{j}}(Q)\,. \label{poly-G-def}
\end{equation}
\item quartic polynomials associated with the vertices of the quadrilateral:
\begin{equation}
r(u)=: {\cal{D}}_{u_i}(h(u,u_i))={\cal D}_{u_j}(\bar{h}(u,u_j))={\cal D}_{u_{ij}}(H(u,u_{ij})) \,.     
\end{equation}  
\end{enumerate}
Furthermore, between the bi-quadratic polynomials the following relations hold on solutions of equation (\ref{gen-eq}).
\begin{equation}
h(u,u_i)\,h(u_j,u_{ij})\,=\,\bar{h}(u,u_j)\,\bar{h}(u_i,u_{ij})\,=\,H(u,u_{ij})\,H(u_i,u_j)\label{hhG-rel}.
\end{equation}
\end{Lem} 

With the above definitions we are now in a position to write an integral form of the three-leg formula for an affine-linear quadrilateral equations of the type 
\eqref{gen-eq} subject to the properties 1)-3) given above. The usual three-leg formula was given in Adler {\it et al.} (2003) on a case-by-case basis. 
In the monograph by Bobenko \& Suris (2008) a result by Adler is cited as an exercise, of which the following is a slight generalization: 
\begin{proposition}
Let $Q$ possess the properties 1)-3) given above. Then the equation
\begin{equation}
\int^{u_i}_{x}\!\!\! \frac{\rd t}{h(u,t)} + \int^{u_{j}}_{y} \!\!\!\frac{\rd t}{\bar{h}(u,t)} + \int^{u_{ij}}_{z(u,x,y)}\!\! \frac{\rd t}{H(u,t)}\,=\,0 \label{3leg-gen}
\end{equation}
holds on solutions of $Q(u,u_i,u_j,u_{ij})=0$, where $x$ and $y$ are arbitrary and where the function $z$ is defined through the relation $Q(u,x,y,z) = 0$.
\end{proposition}
\begin{proof} The proof is straightforward using implicit differentiation with respect to the variables $x$, $y$, $u_i$ and $u_j$.
Substituting $u_{ij}=z(u,u_i,u_j)$, since we require relation (\ref{3leg-gen}) to hold on solutions of $Q(u,u_i,u_j,u_{ij})=0$, and denoting the resulting left-hand side of 
(\ref{3leg-gen}) by $R(u;x,y;u_i,u_j)$, we can easily prove that $R$ is actually independent of $x$ and $y$, and also of $u_i$ and $u_j$. To show this it is sufficient 
to show that $\partial_x R = 0$ through implicit differentiation. Computing the derivative of $R$ with respect to $x$ we obtain: 
\begin{equation}
\frac{\partial R}{\partial x}\,=\,-\,\frac{1}{h(u,x)}\,-\,\frac{1}{H(u,z)}\,\frac{\partial z}{\partial x}\,, \label{R-der-x}
\end{equation}
since $z$ is determined by $u$, $x$ and $y$ through the equation $Q(u,x,y,z)=0$. Implicit differentiation of the latter yields
\begin{equation} 
\frac{\partial z}{\partial x}\,=\,-\,\frac{\partial_x Q(u,x,y,z)}{\partial_z Q(u,x,y,z)}\,=\,-\,\frac{H(u,z)}{h(u,x)}\,, \label{z-der-x}
\end{equation}
which holds on solutions of  $Q(u,x,y,z)=0$ by taking into account the definitions 
$$h(u,x)\,=\,{\cal{D}}_{y,z}(Q(u,x,y,z))\,,\quad H(u,z)\,=\,{\cal{D}}_{x,y}(Q(u,x,y,z))\,. $$
The combination of (\ref{R-der-x}) and (\ref{z-der-x}) leads to the result $\partial_x R = 0$. 
In precisely the same fashion it can be established that $\partial_y R=\partial_{u_i} R=\partial_{u_j} R=0$ , where in the latter two we make use of the 
substitution for $u_{ij}$ as mentioned above. Since as a consequence $R(u;x,y;u_i,u_j)$ does not depend on $x$,$y$, we can chose the latter arbitrarily, leading to the 
conclusion that ~$R(u;x,y;u_i,u_j)=R(u;u_i,u_j;u_i,u_j)$~, where the latter expression vanishes due to the integration limits and taking into account that 
for $x=u_i$, $y=u_j$ we have that $z=u_{ij}$. \end{proof}

Equation (\ref{3leg-gen}) suggests that the three leg form can be written as
\begin{equation}
F(u,u_i,u_j,u_{ij})\,-\,F^0(u)\,=\,0\,, \label{3leg-1}
\end{equation}
where
\begin{equation} \label{H-3leg}
F(u,u_i,u_j,u_{ij})\,:=\, \int^{u_i}\frac{\rd t}{h(u,t)} + \int^{u_j} \frac{\rd t}{\bar{h}(u,t)} + \int^{u_{ij}} \frac{\rd t}{H(u,t)}\,,
\end{equation}
and
\begin{equation} \label{H0-3leg}
F^0(u)\,:=\,F(u,x,y,z)\Big|_{Q(u,x,y,z)=0}\,, 
\end{equation}
where the indefinite integrals denote the primitive (anti-derivative) functions of the integrands with respect to the variables 
$u_i$, $u_j$ and $u_{ij}$, respectively, and the integration constant in each integral is assumed to be independent of the other integration variables. 

Lagrangians for the generalized three-leg formula \eqref{3leg-gen} can be obtained by simple integration over the remaining variable 
$u$, but there is a stronger result that holds. In fact, it turns out that under the assumption of the three-point form of the 
Lagrange function, as was done in Lobb \& Nijhoff (2009), one can actually derive in an almost unique way the form of the Lagrange function
for the affine-linear quadrilateral equations. This is expressed in
\begin{theorem} \label{Lagr-D4-prop}
Let $Q(u,u_i,u_j,u_{ij})$ be an irreducible, affine linear and Kleinian-symmetric polynomial. Then, up to a freedom of having a multiplicative constant, and an 
additive constant, a three-point Lagrangian of the form ${\cal{L}}(u,u_i,u_j)$ of the equation $Q=0$, is given by: 
\begin{eqnarray}
{\cal{L}}(u,u_i,u_j)&=& \int^u\!\!\!\int^{u_i}\!\frac{\rd s \,\rd t}{h(s,t)}\,+\,\int^u\!\!\!\int^{u_j}\!\frac{\rd s \,\rd t}{\bar{h}(s,t)}\,+
\,\int^{u_i}\!\!\!\int^{u_j}  \frac{\rd s \,\rd t}{H(s,t)} \nonumber  \\
& & -\,2\, \int^u\!\!\! F^0(s)\, \rd s . \label{Lagr-D4}
\end{eqnarray}
Integrals denote anti-derivatives of the integrands which are determined up to constants, i.e. independent of $u$, $u_i$ and $u_j$.
\end{theorem}
\begin{proof}
Let us suppose that we are looking for a function ${\cal{L}}(u,u_i,u_j)$ such that
\begin{equation}
\frac{\delta {\cal{L}}}{\delta u}\,\equiv\,\frac{\partial}{\partial u} \Big( {\cal{L}}(u,u_i,u_j) \,+\,{\cal{L}}(u_{-i},u,u_{-i,j})\,+\,{\cal{L}}(u_{-j},u_{i,-j},u) \Big)\,=\,0 \label{EL-eqs-1}
\end{equation}
on solutions of the equation $Q=0$.

First we differentiate the above equation with respect to $u_{-i}$ by taking into account that $u_{-i,j}$ is determined by the equation $Q(u_{-i},u,u_{-i,j},u_j)=0$. This yields
\begin{equation}
\partial_{u_{-i}} \partial_u {\cal{L}}(u_{-i},u,u_{-i,j})\,+\,\partial_{u_{-i,j}} \partial_u {\cal{L}}(u_{-i},u,u_{-i,j}) \frac{\partial u_{-i,j}}{\partial u_{-i}}\,=\,0\,, \label{EL-eqs-2} 
\end{equation}
where
$$\frac{\partial u_{-i,j}}{\partial u_{-i}}\,=\,-\,\frac{\partial_{u_{-i}} Q(u_{-i},u,u_{-i,j},u_j)}{\partial_{u_{-i,j}} Q(u_{-i},u,u_{-i,j},u_j)}\,
=\,-\,\frac{H(u,u_{-i,j})}{h(u,u_{-i})}$$
on solutions of equation $Q(u_{-i},u,u_{-i,j},u_j)=0$. Thus, substituting the latter into (\ref{EL-eqs-2}) and shifting forward in the $i$ direction the resulting relation, one arrives at
\begin{equation}
h(u,u_i)\,\partial_u \partial_{u_i} {\cal{L}}(u,u_i,u_j) \, = \, H(u_i,u_j)\,\partial_{u_i} \partial_{u_j} {\cal{L}}(u,u_i,u_j)\,. \label{EL-eqs-3a}
\end{equation}
In the same fashion, the derivative of (\ref{EL-eqs-1}) with respect to $u_{-j}$ leads to
\begin{equation}
h(u,u_j) \, \partial_u \partial_{u_j} {\cal{L}}(u,u_i,u_j)\,=\, H(u_i,u_j)\,\partial_{u_i} \partial_{u_j} {\cal{L}}(u,u_i,u_j)\,.\label{EL-eqs-3b}
\end{equation}
A third equation follows from the compatibility of (\ref{EL-eqs-3a},\ref{EL-eqs-3b}), namely
\begin{equation}
h(u,u_i)\,\partial_u \partial_{u_i} {\cal{L}}(u,u_i,u_j) \, = \, h(u,u_j) \, \partial_u \partial_{u_j} {\cal{L}}(u,u_i,u_j)\,.\label{EL-eqs-3c}
\end{equation}

The solution of the system of equations (\ref{EL-eqs-3a}--\ref{EL-eqs-3c}) is 
\begin{equation}
{\cal{L}}\,=\,c\,\left(\int^u\!\!\!\int^{u_i}\!\frac{\rd s \,\rd t}{h(s,t)}\,+\,\int^u\!\!\!\int^{u_j}\!\frac{\rd s \,\rd t}{\bar{h}(s,t)}\,+
\,\int^{u_i}\!\!\!\int^{u_j}  \frac{\rd s \,\rd t}{H(s,t)}\,\right)+\,\Phi(u)\,+\,\Psi(u_i)\,+\,{\rm{X}}(u_j)\,, \label{EL-eqs-4}
\end{equation}
where $c$ is a constant, $\Phi$, $\Psi$ and $\rm{X}$ are arbitrary functions of their arguments and double integrals denote anti-derivatives.

To determine the unspecified functions involved in ${\cal{L}}$, we write out explicitly the Euler-Lagrange equations (\ref{EL-eqs-1}) using (\ref{EL-eqs-4}) and taking into account (\ref{3leg-1}--\ref{H0-3leg}). This leads to $2\, c\, F^0(u)+ \Phi^\prime(u) + \Psi^\prime(u) + {\rm{X}}^\prime(u) =0$, where prime denotes differentiation with respect to $u$. From the latter equation, one of the unspecified functions can be defined in terms of the remaining ones. We choose to solve it for $\Phi(u)$ leading to 
$$\Phi(u)\,=\,-\,\Psi(u)\,-\,{\rm{X}}(u)\,-\,2\,c\,\int^u\!\!\!F^0(s) \rd s\,,$$
which is defined up to an additive constant.

Substituting back into (\ref{EL-eqs-4}), it is not difficult to see that the terms involving functions $\Psi$ and $\rm{X}$ can be written as a total difference,
and since their variational derivative does not contribute to the EL equations, they can be omitted from $\cal{L}$. Hence, up to a multiplicative constant and an additive constant, $\cal{L}$ is given by (\ref{Lagr-D4}).
\end{proof}

In all of the above formulas, primitives are determined up to constants. One can fix this freedom during the derivation of the Lagrangian by choosing an arbitrary function $u^0$ along with its shifts as lower limits in the integrals involved in (\ref{EL-eqs-4}). The derivation determines the remaining terms in the latter leading to
\begin{eqnarray}
{\cal{L}}(u,u_i,u_j)&=& \int^u_{u^0}\!\!\int^{u_i}_{u_i^{0}}\!\frac{\rd s \,\rd t}{h(s,t)}\,+\,\int^u_{u^0}\!\!\int^{u_j}_{u^0_j}\!\frac{\rd s \,\rd t}{\bar{h}(s,t)}\,+
\,\int^{u_i}_{u^0_i}\!\!\!\int^{u_j}_{u^0_j}  \frac{\rd s \,\rd t}{H(s,t)} \nonumber  \\
& & -\, \int^{u_i}_{u^0_i}\int^{w(s,u^0,u_{ij}^0)}_{u^0_j}\!\!\!\frac{\rd s \, \rd t}{H(s,t)}\,-\, \int^{u_j}_{u^0_j}\int^{z(t,u_{ij}^0,u^0)}_{u^0_i}\!\!\!\frac{\rd t \, \rd s}{H(s,t)}\,, \label{Lagr-D4-sp}
\end{eqnarray}
where $w$ and $z$ are solutions of the equations $Q(u^0,s,w,u_{ij}^0)=0$ and $Q(u^0,z,t,u_{ij}^0)=0$, respectively.

Lagrangian (\ref{Lagr-D4}) possesses the Klein symmetry in the sense that three different Lagrangians can be derived from this one by applying the interchanges implied by the Klein symmetry.

To make this statement clear, we apply the changes (\ref{gen-D4sym}) to (\ref{Lagr-D4}). The first interchange, i.e. $(u,u_i,u_j,u_{ij}) \rightarrow (u_i,u,u_{ij},u_j)$, which leaves invariant the equation $Q=0$, when applied to $\cal{L}$ leads to
${\cal{L}}(u_i,u,u_{ij})$, which is another Lagrangian of equation $Q=0$. Additionally, applying the changes $(u,u_i,u_j,u_{ij}) \rightarrow (u_j,u_{ij},u_i,u)$ to (\ref{Lagr-D4}) one arrives at another Lagrangian, namely ${\cal{L}}(u_j,u_{ij},u)$. Finally, combining these two symmetries one can derive a third Lagrangian which has the form ${\cal{L}}(u_{ij},u_j,u_i)$.

In Bobenko \& Suris (2010) the proof of the closure relation makes use of a fundamental property of the Lagrangian, which in the context of the present approach can be generalized and understood as the following
\begin{proposition}
All the above Langrangians are defined up to an additive constant and, by choosing it appropriately, their sum is zero on solutions of the equation $Q=0$, i.e.
\begin{equation}
{\cal{L}}(u,u_i,u_j)\,+\,{\cal{L}}(u_i,u,u_{ij})\,+\,{\cal{L}}(u_{ij},u_j,u_i)\,+\,{\cal{L}}(u_j,u_{ij},u)\,=\,0\,. \label{fund-Lagr-D4}
\end{equation}
\end{proposition}

\begin{proof}
To prove that the sum is a constant on solutions of $Q=0$, it is sufficient to prove that its derivatives with respect to $u$, $u_i$ and $u_j$ vanish. Indeed, differentiating it with respect to $u$ or the single shifted values of $u$ and taking into account that $u_{ij}$ depends on $u$, $u_i$ and $u_j$ through the equation $Q=0$, one easily verifies that the sum of the four Lagrangians is constant on solutions of $Q=0$.
\end{proof}

\section{Universal Lagrangian formulation of the ABS equations and Closure} \label{sec-ABS}

Among the members of the class of equations studied in the previous section are the ones possessing the symmetries of the square, or D4 symmetry. These equations may depend on two lattice parameters, denoted by $\alpha_i$ and $\alpha_j$ respectively, and their defining polynomial will be denoted by $Q_{ij}(u,u_i,u_j,u_{ij})$, or simply by $Q_{ij}$, in order to make this dependence explicit.

Since the lattice parameters are assigned to the edges of a plaquette, the symmetries of the square imply that the interchange of $u_i$ and $u_j$ must be followed by a mutual interchange of the lattice parameters, i.e. 
\begin{equation} \label{gen-eqABS} 
Q_{ij}(u,u_i,u_j,u_{ij})\,=\,\epsilon\,Q_{ji}(u,u_j,u_i,u_{ij})\,.
\end{equation}

The ABS equations, which are given in \ref{fGk}, belong to this subclass and depend explicitly on the lattice parameters. They are multidimensionally consistent, cf. Definition \ref{MDC-dis}, and this property depends essentially on the appearance of the lattice parameters in them. These two characteristics of the ABS equations have the following consequences. First of all, the polynomials $h$ assigned to the edges can be factorized as
\begin{eqnarray}
h(u,u_i) \,=\,  \kappa_{ij} h_{i}(u,u_i)\,, \quad h(u_j,u_{ij}) \, =\, \kappa_{ij} h_i(u_j,u_{ij}),\nonumber \\
 h(u,u_j) \, = \, \kappa_{ji}h_j(u,u_j)\,, \quad h(u_i,u_{ij}) \,= \,\kappa_{ji} h_j(u_i,u_{ij}). \label{h-ABS}
\end{eqnarray}
The function $\kappa_{ij} = \kappa(\alpha_i,\alpha_j)$ is antisymmetric, i.e. $\kappa_{ji} = - \kappa_{ij}$, and $h_i(x,y)$ is a symmetric bi-quadratic 
polynomial of $x$ and $y$ and the index $i$ denotes that it depends only on the lattice parameter $\alpha_i$.

Moreover, polynomial $H(x,y)$ is symmetric and bi-quadratic with respect to $x$ and $y$, and symmetric with respect to the lattice parameters. It is more convenient for the analysis of the ABS equations in the following sections to introduce the function
\begin{equation} 
H_{ij}(x,y) \,:=\, \frac{H(x,y)}{\kappa_{ij}}\,, \label{omega-ABS}
\end{equation}
which is antisymmetric with respect to the interchange of indices : $H_{ij}=-H_{ji}$. Again, the indices denote the dependence on the lattice parameters and not shifts in the corresponding lattice directions. Functions $h_i$, $H$ and $\kappa_{ij}$ can be found also in \ref{fGk}.

Finally, the factorization of polynomials $h$ implies that the three leg form for these equations can be written as
\begin{equation}
\int^{u_i}_{x}\!\!\! \frac{\rd t}{h_i(u,t)} - \int^{u_{j}}_{y}\!\!\! \frac{\rd t}{h_j(u,t)} + \int^{u_{ij}}_{z(u,x,y)}\!\! \frac{\rd t}{H_{ij}(u,t)}\,=\,0\,, \label{3leg}
\end{equation}
where $x$, $y$ are arbitrary and $z$ is any solution of $Q_{ij}(u,x,y,z) = 0$.

In this case, the function $F$ (\ref{H-3leg}) is antisymmetric with respect to the interchanges of the indices and is given by
\begin{equation} \label{ABS-H-3leg}
F_{ij}\,\equiv\,F(u,u_i,u_j,u_{ij};\alpha_i,\alpha_j)\,:=\,\int^{u_i}\!\!\! \frac{\rd t}{h_i(u,t)}\,-\,\int^{u_j} \!\!\!\frac{\rd t}{h_j(u,t)} \,+
\,\int^{u_{ij}}\!\!\! \frac{\rd t}{H_{ij}(u,t)}\,.
\end{equation}
Similarly, the function $F^0$ (\ref{H0-3leg}) takes the following form
\begin{equation} \label{ABS-H0-3leg}
F_{ij}^0(u) \,:=\, F(u,x,y,z;\alpha_i,\alpha_j) \Big|_{Q_{ij}(u,x,y,z)=0}\,. \end{equation}

Having introduced through Proposition \ref{Lagr-D4-prop} a Lagrangian for the ABS equations, we now proceed and prove that it obeys the closure property. 
\begin{theorem} \label{Lang-clos-rel-ABS}
The following function
\begin{eqnarray}
{\cal{L}}_{ij}(u,u_i,u_j)&=& \int^u\!\!\!\int^{u_i}\!\frac{\rd s \,\rd t}{h_i(s,t)}-\int^u\!\!\!\int^{u_j}\!\frac{\rd s \,\rd t}{h_j(s,t)}+
\int^{u_i}\!\!\!\int^{u_j}  \frac{\rd s \,\rd t}{H_{ij}(s,t)}\nonumber\\
&& -\,\int^{u_i}\!\!\!F^0_{ij}(t) \rd t\,-\,\int^{u_j}\!\!\!F^0_{ij}(t) \rd t \label{disLag1}
\end{eqnarray}
defines a Lagrangian for the ABS equation $Q_{ij}=0$, meaning that the EL equations are satisfied whenever $Q_{ij}(u,u_i,u_j,u_{ij})=0$.
Moreover, the Lagrangians ${\cal{L}}_{ij}$, ${\cal{L}}_{j \ell}$ and ${\cal{L}}_{\ell i}$ satisfy the closure relation
\begin{equation}
\Delta_i {\cal{L}}_{j \ell}\,+\,\Delta_j {\cal{L}}_{\ell i} \,+\, \Delta_\ell {\cal{L}}_{ij}\,=\,0\,,\label{dis-clos}
\end{equation}
which holds on solutions of the ABS equation $Q_{ij}=0$, $Q_{j\ell}=0$ and $Q_{\ell i}=0$, and hence, in view of skew symmetry of the 
Lagrangian functions ${\cal L}_{ij}(u,u_i,u_j)=-{\cal L}_{ji}(u,u_j,u_i)$, it defines a closed discrete 2-form.
\end{theorem}

For the proof of the closure relation, we will need the following
\begin{lemma} \label{f-omega-relations}
Let $h_i(u,x)$ be any of the polynomials related to the ABS equation $Q_{ij}=0$. Then, the relation
\begin{equation}
2\,\partial_{\alpha_i} h_i(u,x) \,+\,{\cal{D}}_{u,x}(h_i(u,x))\,=\,0 \label{h-a-der}
\end{equation}
holds identically.

Similarly, if $H_{ij}(u,x)$ is any of the diagonal polynomials corresponding to the ABS equation $Q_{ij}=0$, then the following identities hold
\begin{equation}
2\, \partial_{\alpha_i} H_{ij}(u,x)- {\cal{D}}_{u,x}(H_{ij}(u,x)) = 0,\quad 2\, \partial_{\alpha_j} H_{ij}(u,x) + 
{\cal{D}}_{u,x}(H_{ij}(u,x))= 0. \label{omega-a-der}
\end{equation}
\end{lemma}

The identity for $h_i$ follows from the symmetry analysis of the ABS equations, (Xenitidis 2009, Mikhailov {\it{et al.}} 2010), and its derivation is given in  \ref{aDerivative}. Actually, this relation expresses the fact that the lattice parameters play the role of 
the continuous ``master symmetry'' parameters. The same relation can be found in the context of integrable chains of differential--difference equations and their master symmetries (Yamilov 2006).

The relation for $H_{ij}$ can be proven either directly, or, in the case of Q1--Q4, using the observation that
$ H_{ij}(u,x)\,\equiv\,-h(u,x,\alpha_i-\alpha_j)$.

We now proceed with the proof of Theorem \ref{Lang-clos-rel-ABS}. 

\begin{proof}
Since the ABS equations belong to a subclass of the family of equations studied in the previous section, the Lagrangian (\ref{disLag1}) easily follows from (\ref{EL-eqs-4}) by assuming it to be symmetric with respect to the interchange of indices, choosing the multiplicative constant to be $\kappa_{ij}$ and using relations (\ref{h-ABS}, \ref{omega-ABS}).

Let us now denote the left hand side of (\ref{dis-clos}) by $\cal{C}$. To prove that the closure relation holds on solutions of equations $Q_{ij}=0$, $Q_{j\ell}=0$ and $Q_{\ell i}=0$, it is sufficient to prove that $\cal{C}$ is independent of $u$ and its single shifts, as well as the lattice parameters. Observing that the closure relation 
is symmetric with respect to any permutation of the indices, it is sufficient to show that $\partial_u {\cal{C}} =  \partial_{u_i} {\cal{C}} = \partial_{\alpha_i}{\cal{C}}=0$.

Actually, it can be proven that {\sl{all}} the first order derivatives of $\cal{C}$ vanish on solutions of $Q_{ij}=0$, $Q_{j\ell}=0$ and $Q_{\ell i}=0$. First of all, writing $\cal{C}$ out explicitly, it follows that it is independent of $u$. Now, on the one hand, the derivative of $\cal{C}$ with respect to $u_i$ reads as follows 
\begin{equation} \partial_{u_i}{\cal{C}}\,=\,  -T_{ij}(u_i)-T_{\ell i}(u_i)\,,\label{deriv-clos-ui} \end{equation}
where 
$$T_{ij}(u_i)\,:=\, F_{ij}(u,u_i,u_j,u_{ij})\,-\,F^0_{ij}(u_i).$$
Obviously, in view of $Q_{ij}=0$ and $Q_{\ell i}=0$, relation (\ref{deriv-clos-ui}) implies that $\partial_{u_i} {\cal{C}}=0$.

On the other hand, the derivative of $\cal{C}$ with respect to $u_{ij}$ has the form
\begin{equation}
\partial_{u_{ij}} {\cal{C}} \,=\, {\cal{S}}_i (T_{j\ell}(u_j)) + {\cal{S}}_j(T_{\ell i}(u_i))\,=\,0\,, \label{deriv-clos-uij}
\end{equation}
where the last equality holds on $Q_{j\ell}=0$ and $Q_{\ell i}=0$ and their shifts. Thus, $\cal{C}$, due to its symmetry, is independent of $u$ and all of its shifts.

The next step is to prove that the closure relation is independent of the lattice parameters. Considering the derivative of $\cal{C}$ with respect to $\alpha_i$ and using Lemma (\ref{f-omega-relations}), one arrives at
$$
\partial_{\alpha_i} {\cal{C}} \,=\,\frac{1}{2}\,\ln \left(\frac{H_{\ell i}(u,u_{i,\ell}) H_{\ell i}(u_{j \ell},u_{ij})}{H_{i j}(u,u_{ij}) H_{i j}(u_{i,\ell},u_{j \ell})} \right) - \frac{\partial I}{\partial \alpha_i}\,,
$$
where
$$ I := \Delta_\ell \left(\int^{u_i}\!\!\!F_{ij}^0(s) \rd s + \int^{u_j}\!\!\!F_{ij}^0(s) \rd s \right)+ \Delta_j \left(\int^{u_\ell}\!\!\!F_{\ell i}^0(s) \rd s + \int^{u_i}\!\!\!F_{\ell i}^0(s) \rd s \right)\,.
$$

It is not difficult to verify case by case that
$$ H_{\ell i}(u,u_{i \ell}) H_{\ell i}(u_{j \ell},u_{ij})\,-\, H_{i j}(u,u_{ij}) H_{i j}(u_{i \ell},u_{j \ell})\Big|_{Q_{ij}=0,Q_{j\ell}=0,Q_{\ell i}=0}=0\,,$$
which is equivalent to the affine linear equation $P(u,u_{ij},u_{j \ell},u_{i \ell})=0 $ relating the four points $u$, $u_{ij}$, $u_{j \ell}$ and $u_{i \ell}$.

Finally, the definition (\ref{ABS-H0-3leg}) of $F^0_{ij}$ and relations\footnote{These relations hold on solutions of the equation $Q_{ij}=0$ and can be verified case by case since their proof is lengthy and is omitted here.} 
\begin{equation}
\frac{\partial u_{ij}}{\partial \alpha_i} = \frac{-H_{ij}(u,u_{ij})}{2}\, \partial_u \ln\left( \frac{h_i(u,u_i)}{H_{ij}(u,u_{ij})}\right),\, \frac{\partial u_{ij}}{\partial \alpha_i}=\frac{-H_{ij}(u_i,u_{j})}{2}\, \partial_{u_j}\left( \frac{h_i(u_j,u_{ij})}{H_{ij}(u_i,u_{j})}\right), \label{uij-a-der}
\end{equation}
imply that $\partial_{\alpha_i} I = 0$. Combining all these it follows that $\partial_{\alpha_i} {\cal{C}}=0$.

Hence, the closure relation is constant on solutions of the corresponding ABS equation. We will show now that this constant is zero. The left hand side of the closure relation, when is evaluated on solutions of the corresponding ABS equation, becomes
\begin{eqnarray}
{\cal{C}}\,=\, -\int^{u_{j \ell}}\!\!\!\int^{u_{i \ell}}\!\!\!\frac{\rd s \, \rd t}{H_{ij}(s,t)} + \int^{u_i}\!\!\!\int^{u_j}\!\!\!\frac{\rd s \,\rd t}{H_{ij}(s,t)}+ \int^{u_j}\!\!\!\int^{u_{ij}}\!\!\!\frac{\rd s \,\rd t}{h_{i}(s,t)} - \int^{u_\ell}\!\!\!\int^{u_{i \ell}}\!\!\!\frac{\rd s \, \rd t}{h_{i}(s,t)} \nonumber\\
 \nonumber \\
 +\,\,\, {\mbox{terms obtained by cycling permuting indices}}\,. \hspace{4.8cm} \label{C-on-solutions}
\end{eqnarray}
This formula suggests the existence of the solution\footnote{This solution was considered also in Bobenko \& Suris (2010).}
\begin{equation}
u_{ij}\,=\,u_\ell\,,\quad u_{j \ell}\,=\,u_i\,,\quad u_{i \ell}\,=\,u_j\,, \label{special-solution}
\end{equation}
which, when combined with relation (\ref{C-on-solutions}), leads to ${\cal{C}}=0$, implying that $\cal{C}$ vanishes on any solution of the corresponding ABS equation.

The particular solution (\ref{special-solution}) is equivalent to the system
$$Q_{ij}(u,u_i,u_j,u_\ell)\,=\,0\,,\quad Q_{j\ell}(u,u_j,u_\ell,u_i)\,=\,0\,,\quad Q_{\ell i}(u,u_\ell,u_i,u_j)\,=\,0\,.$$
To prove that such a solution exists, it is sufficient to prove that two of the equations imply the third one. Indeed, writing the two first equations in terms of the polynomials $h$, i.e.
$$h_i(u,u_i) h_i(u_j,u_\ell)\,=\,h_j(u,u_j) h_j(u_i,u_\ell),\,\, h_j(u,u_j) h_j(u_\ell,u_i)\,=\,h_\ell(u,u_\ell) h_\ell(u_j,u_i),$$
the elimination of polynomials $h_j$ between them yields \footnote{Here we have used the fact that polynomials $h_i$ are symmetric, i.e. $h_i(x,y)=h_i(y,x)$.}
$$h_i(u,u_i)\,h_i(u_j,u_\ell)\,=\, h_\ell(u,u_\ell)\, h_\ell(u_j,u_i)\,,$$
which is equivalent to the third equation of the above system.
\end{proof}

\noindent {\bf{Remark 3.1.}}
There exist deformations of the H equations in the ABS list, which were first presented in Adler {\it{et al.}} (2009), and further studied in Xenitidis \& Papageorgiou (2009). The defining functions of these equations possess the basic properties of affine linearity and irreducibility, and the symmetries of the rhombus. They are defined on a black-white lattice which implies that they can be written in non-autonomous form. In this formulation, the polynomials assigned to the edges have the same form and can be factorized like the polynomials (\ref{h-ABS}) corresponding to the ABS equations. On the other hand, there are two diagonal polynomials $H_1$ and $H_2$ related to each other by shifts, e.g. ${\cal{S}}_i H_1(u_i,u_{j}) = H_2(u_{ii},u_{ij})$, cf. Xenitidis \& Papageorgiou (2009). Actually, from Theorem \ref{Lang-clos-rel-ABS} follows a Lagrangian for the deformed H equations which satisfies the closure relation as well. One has to use the polynomials $h$ corresponding to the deformed equations and replace $H_{ij}(u,u_{ij})$ and $H_{ij}(u_i,u_j)$ by ${H_1}_{ij}(u_i,u_j)$ and ${H_2}_{ij}(u,u_{ij})$, respectively, in (\ref{disLag1}). \hfill $\Box$ \\

\noindent {\bf{Remark 3.2.}}
The lift of the ABS equation $Q_{ij}=0$ to a two-field discrete system was suggested by Papageorgiou \& Tongas (2009) in the construction of Yang-Baxter maps. Denoting by $u$ and $v$ the fields involved in this system, the latter has the following form
\begin{equation}
Q_{ij}(u,v_i,u_j,u_{ij})\,=\,0\,,\quad Q_{ij}(v,v_i,u_j,v_{ij})\,=\,0\,. \label{lift-ABS}
\end{equation}
A Lagrangian for the above system follows actually from (\ref{disLag1}) and it has the following form
\begin{eqnarray*}
{\cal{L}}(u,v,v_i,u_j) &=& \int^u\!\!\!\int^{v_i}\!\!\!\frac{\rd s \,\rd t}{h_i(s,t)}\,-\,\int^v\!\!\!\int^{u_j}\!\!\!\frac{\rd s \,\rd t}{h_j(s,t)}\,+\,\int^{v_i}\!\!\!\int^{u_j}\!\!\!\frac{\rd s \,\rd t}{H_{ij}(s,t)}\\
 &&-\,\int^{v_i}\!\!\!F^0_{ij}(s) \rd s \,-\,\int^{u_j}\!\!\!F^0_{ij}(s) \rd s\,,
\end{eqnarray*}
where $F^0_{ij}$ is defined in (\ref{ABS-H0-3leg}). It is a straightforward calculation to prove that the variational derivatives of ${\cal{L}}(u,v,v_i,u_j)$ with respect to $u$ and $v$ vanish on solutions of $Q_{ij}(v_{-j},v_{i,-j},u,v_{i})=0$ and $Q_{ij}(u_{-i},v,u_{-i,j},u_{j})\,=\,0$, respectively.\hfill $\Box$

\section{Multidimensionally consistent continuous systems from symmetry reductions} \label{sec-Sigma}

All of the ABS equations admit a pair of extended generalized symmetries which play the role of master symmetries for the corresponding generalized symmetries (Rasin \& Hydon 2008, Xenitidis 2009). This pair of extended symmetries can be given as the differential operators
\begin{equation}
V_i =n_i \left(\frac{h_i(u,u_i)}{u_i-u_{-i}}-\frac{1}{2} \partial_{u_i}h_i(u,u_i)\right)\partial_u-\partial_{\alpha_i}\,.
\label{master-sym}
\end{equation}

Since the equations under consideration depend explicitly on the lattice parameters, the same holds for their solutions. Thus, it is interesting to consider solutions of 
the ABS equations remaining invariant under the action of the extended symmetries.
Such solutions satisfy the discrete equation under consideration and the system of differential-difference equations
\begin{equation}
\frac{\partial u}{\partial \alpha_i}\,+\,n_i\,\left(\frac{h_i(u,u_i)}{u_i-u_{-i}}\,-\,\frac{1}{2} \partial_{u_i}h_i(u,u_i)\right)\,=\,0\,,
\label{dif-dif-eq}
\end{equation}
From this overdetermined system Tsoubelis \& Xenitidis (2009) derived the following {\sl{system of partial differential equations}}
\begin{eqnarray}
\frac{\partial u_i}{\partial \alpha_j} &=& \frac{H_{ij}}{h_j}\,\frac{\partial u}{\partial \alpha_j}\,+\,n_j \,\nabla_j H_{ij}\,, \label{sys1} \\
\frac{\partial u_j}{\partial \alpha_i} &=& \frac{H_{ji}}{h_i}\,\frac{\partial u}{\partial \alpha_i}\,+\,n_i \,\nabla_i H_{ji}\,, \label{sys1a} \\
\frac{\partial^2 u}{\partial \alpha_i \partial \alpha_j} &=& A_{ij} \frac{\partial u}{\partial \alpha_i} \frac{\partial u}{\partial \alpha_j} + B_{ij} 
\frac{\partial u}{\partial \alpha_i} + B_{ji} \frac{\partial u}{\partial \alpha_j} + \Gamma_{ij}\,, \label{sys2}
\end{eqnarray}
where the arguments of $h_i(u,u_i)$, $h_j(u,u_j)$ and $H_{ij}(u_i,u_j)$ have been omitted and $n_i$, $n_j$ were scaled by 2. Moreover,
\begin{equation}
\nabla_i X\,:=\, \frac{\partial X}{\partial u_i}\, -\, \left( \partial_{u_i} \ln h_i \right)\, X \,=\, h_i\,\frac{\partial}{\partial u_i}\left( \frac{X}{h_i} \right)\,, \label{conne}
\end{equation}
and the coefficients in the last equation of the system are given by the following relations.
\begin{eqnarray}
&& A_{ij} \,:=\, \partial_u \ln h_i - \frac{1}{h_j} \nabla_i\, H_{ij} = \partial_u \ln h_j + \frac{1}{h_i} \nabla_j\, H_{ij}\,, \label{Aij} \\
&& B_{ij} \,:=\,-n_j\, \nabla_i \nabla_j \,H_{ij} = n_j \, \left(- h_j\, \partial_u \nabla_j 1 + \frac{1}{h_i} \nabla_j (h_j \,\nabla_j\, H_{ij}) \right),\quad  \label{Bij} \\
&&\Gamma_{ij} \,:=\, -n_i\, n_j\,\nabla_i (h_i \nabla_i \nabla_j\, H_{ij}) \,=\,n_i\, n_j \, \nabla_j (h_j \nabla_i \nabla_j \,H_{ij})\,, \label{Gij}
\end{eqnarray}
In what follows, system (\ref{sys1}--\ref{Gij}) will be denoted as $\Sigma_{ij}$. \\

\noindent {\bf{Remark 4.1.}} The derivation of $\Sigma_{ij}$ implies that $u_i$ and $u_j$ are related to $u$ by shifts. Actually, this relation can be forgotten and all the functions involved in $\Sigma_{ij}$ can be treated as independent, which is the point of view we adopt from now on: we consider $\Sigma_{ij}$ as a closed-form coupled system for the three functions $u$, $u_i$ and $u_j$. \hfill $\Box$ \\

\noindent {\bf{Remark 4.2.}} The two equivalent forms at the definitions (\ref{Aij})-(\ref{Gij}) of the coefficients $A_{ij}$, $B_{ij}$ and $\Gamma_{ij}$, respectively, follow from the two different 
ways which can be used to derive the second equation of the system and the fact that the master symmetries commute (Xenitidis 2009). The latter implies that these equalities 
are identities. In fact, the basic identity is relation (\ref{Aij}) from which (\ref{Bij},\ref{Gij}) follow. Specifically, the $\nabla_i$-derivative of (\ref{Aij}) leads to (\ref{Bij}). On the 
other hand, multiplying the latter by $h_i$ and taking the $\nabla_i$-derivative of the resulting equation, one arrives at the third identity (\ref{Gij}). \hfill $\Box$ \\

Some useful general properties of the operator $\nabla_i$ defined in (\ref{conne}) can be derived directly from its definition and are given in the following 
statement. 
\begin{proposition} \label{identities-nabla}
Let $h_i(u,u_i)$, $h_j(u,u_j)$ be any pair of smooth functions and $\nabla_i$ be defined as in (\ref{conne}). Then the following identities hold. \\
{\it{ i)}} $\nabla_i h_i \,=\,0$ identically; \\ 
{\it{ii)}} $\left[ \nabla_i\,,\, \nabla_j\right] X\,=\,0$, for any function $X$;\\  
{\it{ iii)}} $\partial_u \nabla_i \nabla_j Y + (\partial_u \partial_{u_i} \ln h_i) \nabla_j Y + (\partial_u \partial_{u_j} \ln h_j) \nabla_i Y \,=\,0$, for any function $Y$ independent of $u$.
\end{proposition}

An important property of $\Sigma_{ij}$ is its multidimensional consistency, i.e. the compatibility conditions among the equations of systems $\Sigma_{ij}$, $\Sigma_{j\ell}$ and $\Sigma_{\ell i}$ lead to no further restrictions. To make more precise what we mean by a multidimensionally consistent system of PDEs, we give the following definition in a way parallel to Definition \ref{MDC-dis} of Section \ref{sec-Klein} for the discrete case. 
\begin{Def} 
Let $F(u,\partial_{\alpha_i}u,\partial_{\alpha_j}u,\partial_{\alpha_i} \partial_{\alpha_j}u)=0$ denote a system of differential equations.
Imposing the dependence of $u$ on a third continuous variable $\alpha_k$ such that 
$$F(u,\partial_{\alpha_j}u,\partial_{\alpha_k}u,\partial_{\alpha_j} \partial_{\alpha_k}u)=0\,,\quad 
F(u,\partial_{\alpha_k}u,\partial_{\alpha_i}u,\partial_{\alpha_k} \partial_{\alpha_i}u)=0\,, $$
if the three different ways to evaluate $\partial_{\alpha_i} \partial_{\alpha_j} \partial_{\alpha_k} u$ lead to the same result, then we call 
the system of PDEs defined by $F$ multidimensionally consistent. 
\end{Def} 

The multidimensional consistency of $\Sigma_{ij}$ can be checked by straightforward calculations and is equivalent to the identity
\begin{eqnarray}
T_{ij\ell}&:=&\left(H_{j\ell} \partial_{u_\ell} H_{\ell i}- H_{\ell i} \partial_{u_\ell} H_{j\ell}\right) + 
 \left(H_{\ell i} \partial_{u_i} H_{ij}-H_{ij} \partial_{u_i} H_{\ell i}\right) \nonumber \\
  &&+ \left(H_{ij} \partial_{u_j} H_{j\ell}-H_{j \ell} \partial_{u_j} H_{ij}\right)\,\equiv\,0\,, \label{relG}
\end{eqnarray}
which can be checked case by case for all the ABS equations and their corresponding functions $H_{ij}$.

To see how (\ref{relG}) expresses the multidimensional consistency of the ABS equations and their continuous counterparts $\Sigma_{ij}$, consider the compatibility condition
$$\frac{\partial}{\partial \alpha_\ell }\left( \frac{\partial u_i}{\partial \alpha_j} \right)\,-\,\frac{\partial}{\partial \alpha_j }\left( \frac{\partial u_i}{\partial \alpha_\ell} \right)\,=\,0\,,$$
which must hold on the solutions of $\Sigma_{ij}$, $\Sigma_{j\ell}$ and $\Sigma_{\ell i}$. The latter relation, when it is written out explicitly using the equations constituting these systems, leads to
\begin{equation*}
\frac{1}{h_j h_\ell}\,\left(  \frac{\partial u}{\partial \alpha_j}\frac{\partial u}{\partial \alpha_\ell} + \frac{n_\ell}{2} \frac{\partial u}{\partial \alpha_j} h_\ell \nabla_\ell + \frac{n_j}{2} \frac{\partial u}{\partial \alpha_\ell} h_j \nabla_j + \frac{n_j n_\ell}{4} h_j h_\ell \nabla_j \nabla_\ell\right) T_{ij\ell} \,=\,0\,, \label{compat-cond-1}
\end{equation*}
which holds in view of (\ref{relG}). In the same fashion, the three different ways to evaluate $\partial_{\alpha_i}\partial_{\alpha_j}\partial_{\alpha_k} u$ and their consistency lead to similar expressions which are satisfied whenever $T_{ijk} \equiv 0$. In the following section where Lagrangian functions will be constructed for the system $\Sigma_{ij}$ we will see that the latter identity also ensures that these Lagrangians satisfy the closure relation. \\

\noindent {\bf{Remark 4.3.}} As for the discrete ABS equations, where a Lax pair follows from multidimensional consistency (Bobenko \& Suris 2002, Nijhoff 2002), the multidimensional consistency of $\Sigma_{ij}$ can be used to derive an auto-B{\"a}cklund transformation and a Lax pair for this system (Tsoubelis \& Xenitidis 2009). To present them, it is more convenient to introduce the following notation:
$$\tilde{H}_i \,:=\, \frac{H(u_i,{\tilde{u}},\alpha_i,\lambda)}{\kappa(\alpha_i,\lambda)}\,,$$
i.e. $H_{ij}$ with $u_j$ and $\alpha_j$ replaced by $\tilde{u}$ and $\lambda$, respectively. Now, the auto-B{\"a}cklund transformation can be written as
\begin{eqnarray*}
&& \frac{\partial \tilde{u}}{\partial \alpha_i}\,=\,\frac{\tilde{H}_i}{h_i}\,\frac{\partial u}{\partial \alpha_i}\,+\, n_i\, \nabla_i \tilde{H}_i \,,\quad  \frac{\partial \tilde{u}}{\partial \alpha_j}\,=\,\frac{\tilde{H}_j}{h_j}\,\frac{\partial u}{\partial \alpha_j}\,+\, n_j\, \nabla_j \tilde{H}_j \,,\\ 
&& Q(u,u_i,\tilde{u},\tilde{u}_i;\alpha_i,\lambda)\,=\,0 \,, \quad \quad Q(u,u_j,\tilde{u},\tilde{u}_j;\alpha_j,\lambda)\,=\,0 \,,
\end{eqnarray*}
where the last two equations coincide with the discrete equation related to $\Sigma_{ij}$.

Moreover, a Lax pair for $\Sigma_{ij}$ has the following form
$$\frac{\partial \Phi}{\partial \alpha_i} \,=\, {\rm{M}}_{i}\, \Phi\,,\quad \frac{\partial \Phi}{\partial \alpha_j} \,=\, {\rm{M}}_{j}\, \Phi $$
where 
$${\rm{M}}_{i}\,:=\,\left(\begin{array}{cc}
-\frac{1}{2} A^\prime _i & -\frac{1}{2} A^{\prime\prime}_i \\
A_i & \frac{1}{2} A_i^\prime \end{array} \right)\,,$$
with
$$A_i \,:=\,\frac{\tilde{H}_i}{ h_i}\,\frac{\partial u}{\partial \alpha_i}\,+\, n_i \, \nabla_i \tilde{H}_i $$
and prime denoting differentiation with respect to $\tilde{u}$. Also, in the matrices ${\rm{M}}_i$, the function $A_i$ and its derivatives are evaluated at $\tilde{u}=0$. \hfill $\Box$

\section{Lagrangian formulation of the differential-difference and PDE system} \label{sec-con-Lagr}

In this section, we will show that there exist Lagrange structures for the differential-difference equations (\ref{dif-dif-eq}) and for the PDE system $\Sigma_{ij}$, where 
in addition the latter satisfies a closure relation. We will start by presenting the Lagrangian for the system defined by the differential-difference equation \eqref{dif-dif-eq}. 

\begin{Pro}
The following action functional 
\begin{eqnarray} 
&& S[u(n_i;\alpha_i)] = \sum_{n_i}\int \left\{ \frac{\partial u_{-i}}{\partial\alpha_i}\int^u_{u^0} \frac{dx}{h_{i}(u_{-i},x)}   
-\frac{\partial u^0_{i}}{\partial\alpha_i}\int^u_{u^0} \frac{dx}{h_{i}(x,u^0_{i})} \right. \nonumber \\ 
&&\qquad  + \left. n_i\left[\ln(u_i-u_{-i})-\frac{1}{2}\ln (h_i(u,u_i)) \right]-\frac{1}{2}\ln(h_i(u,u_i^0))\right\} \rd \alpha_i   \label{dif-dif-act} 
\end{eqnarray} 
forms an action for the differential-difference system \eqref{dif-dif-eq} in the sense that the EL equations 
are satisfied if $u=u(n_i,\alpha_i)$ is a solution of the latter. 
\end{Pro} 
\begin{proof} The proof is by direct computation, and uses crucially the relation \eqref{h-a-der}, as well as the identity 
$$ \frac{h_i(u,u_i)}{u_i-u_{-i}}\,-\,\frac{1}{2} \partial_{u_i}h_i(u,u_i)=\frac{h_i(u,u_{-i})}{u_i-u_{-i}}\,+\,\frac{1}{2} \partial_{u_{-i}}h_i(u,u_{-i})\,, 
$$
which holds for general symmetric bi-quadratics $h_i$.  Thus, we obtain the EL equations in the following form:  
\begin{eqnarray}
&&\frac{\partial{\cal L}}{\partial u} = \frac{1}{h_i(u,u_{-i})}\frac{\partial u_{-i}}{\partial\alpha_i}-\frac{1}{h_{i}(u,u_{i})}
\frac{\partial u_{i}}{\partial\alpha_i} \nonumber \\ 
&&\quad +\frac{n_i+1}{u-u_{i,i}}+\frac{n_i-1}{u-u_{-i,-i}}-\frac{1}{2}\frac{n_i+1}{h_i(u,u_i)}\partial_u h_i(u,u_i)
-\frac{1}{2}\frac{n_i-1}{h_{i}(u,u_{-i})}\partial_u h_{i}(u,u_{-i})   \nonumber 
\end{eqnarray}
where ${\cal L}$ denotes the expression between curly brackets on the right-hand side of \eqref{dif-dif-act}, 
which is a combination of two copies of the differential-difference equation \eqref{dif-dif-eq} one shifted forward, and one shifted backward in the 
lattice variable $n_i$. 
\end{proof}

The multidimensional consistency of $\Sigma_{ij}$ can be established on the level of its Lagrangian formulation by proving the closure relation for the corresponding Lagrangian 2-form structure in accordance with its discrete counterpart.

\begin{theorem}
A Lagrangian of system $\Sigma_{ij}$ is given by
\begin{eqnarray}\label{sys-Lag}
&& {\cal{L}}_{ij} = \frac{1}{h_i}\frac{\partial u}{\partial \alpha_i} \frac{\partial u_i}{\partial \alpha_j}-\frac{1}{h_j}\frac{\partial u}{\partial \alpha_j} \frac{\partial u_j}{\partial \alpha_i}-\frac{H_{ij}}{h_i\, h_j}\frac{\partial u}{\partial \alpha_i} \frac{\partial u}{\partial \alpha_j} - n_i  n_j\, \nabla_i \nabla_j \,H_{ij} \\
&&{\phantom{{\cal{L}}_{ij} =}} - n_j  \left(\frac{1}{h_i} \nabla_j \, H_{ij} + \partial_u \ln h_j \right)\,\frac{\partial u}{\partial \alpha_i}   \nonumber  +  n_i \left(\,\frac{1}{h_j} \nabla_i \,H_{j i} + \partial_u \ln h_i \right)\,\frac{\partial u}{\partial \alpha_j}\nonumber
 \end{eqnarray}
and is antisymmetric: ${\cal{L}}_{ji} \,=\, -\,{\cal{L}}_{ij}$.

Moreover, Lagrangians ${\cal{L}}_{ij}$, ${\cal{L}}_{j\ell}$ and ${\cal{L}}_{\ell i}$ satisfy the closure relation
\begin{equation}
{\rm{D}}_{\alpha_i} {\cal{L}}_{j \ell}\,+\,{\rm{D}}_{\alpha_j} {\cal{L}}_{\ell i}\,+\,{\rm{D}}_{\alpha_\ell} {\cal{L}}_{ij}\,=\,0 \label{closure-cont}
\end{equation}
on solutions of systems $\Sigma_{ij}$, $\Sigma_{j \ell}$ and $\Sigma_{\ell i}$.
\end{theorem} 

\begin{proof}
The variational derivatives of ${\cal{L}}_{ij}$ with respect to $u_i$ and $u_j$ yield the second equation of $\Sigma_{ij}$, while its variational derivative with respect to $u$ leads to a linear combination of the two first equations of $\Sigma_{ij}$. Indeed, consider the variational derivative of ${\cal{L}}_{ij}$ with respect to $u_i$, i.e.
$$\frac{\delta {\cal{L}}_{ij}}{\delta u_i} = \frac{\partial}{\partial \alpha_j}\left( \frac{1}{h_i}\frac{\partial u}{\partial \alpha_i} \right) - \frac{\partial {\cal{L}}_{ij}}{\partial u_i}\,. $$
Writing out explicitly the above relation, the terms involving the derivative of $u_i$ with respect to $\alpha_j$ cancel out. Factorizing the remaining terms, one arrives at 
$$\frac{\delta {\cal{L}}_{ij}}{\delta u_i} \,=\, \frac{1}{h_i} \left(\frac{\partial^2 u}{\partial \alpha_i \partial \alpha_j} - A_{ij} \frac{\partial u}{\partial \alpha_i} \frac{\partial u}{\partial \alpha_j} - B_{ij} \frac{\partial u}{\partial \alpha_i} - B_{ji} \frac{\partial u}{\partial \alpha_j} - \Gamma_{ij} \right)\,,$$
which is identically zero in view of the third equation of $\Sigma_{ij}$.

Similarly, the variational derivative of the Lagrangian with respect to $u_j$ takes the form
$$\frac{\delta {\cal{L}}_{ij}}{\delta u_j} \,=\, \frac{-1}{h_j} \left(\frac{\partial^2 u}{\partial \alpha_i \partial \alpha_j} - A_{ij} \frac{\partial u}{\partial \alpha_i} \frac{\partial u}{\partial \alpha_j} - B_{ij} \frac{\partial u}{\partial \alpha_i} - B_{ji} \frac{\partial u}{\partial \alpha_j} - \Gamma_{ij} \right)\,=\,0\,.$$

For the variational derivative with respect to $u$, we have
\begin{eqnarray*}
\frac{\delta {\cal{L}}_{ij}}{\delta u} &=& \frac{\partial}{\partial \alpha_i} \left\{\frac{1}{h_i} \left(\frac{\partial u_i}{\partial \alpha_j} - \frac{H_{ij}}{h_j}\,\frac{\partial u}{\partial \alpha_j}\,-\,n_j \,\nabla_j \,H_{ij} \right) - n_j \,\partial_u \ln h_j \right\}\\
&+&\frac{\partial}{\partial \alpha_j} \left\{\frac{-1}{h_j} \left(\frac{\partial u_j}{\partial \alpha_i} + \frac{H_{ij}}{h_i}\,\frac{\partial u}{\partial \alpha_i}\,-\,n_i \,\nabla_i \,H_{ij} \right) + n_i\,\partial_u \ln h_i \right\}- \frac{\partial {\cal{L}}_{ij}}{\partial u}\,.
\end{eqnarray*}
The terms in the parentheses are the two first equations of the system, thus the above equation becomes
$$
\frac{\delta {\cal{L}}_{ij}}{\delta u} \,=\, \frac{\partial}{\partial \alpha_i} \left(- n_j \partial_u \ln h_j \right) + \frac{\partial}{\partial \alpha_j} \left(n_i  \partial_u \ln h_i \right) \,-\, \frac{\partial {\cal{L}}_{ij}}{\partial u}\,.
$$
Writing out explicitly the right hand side of the last relation, we use the two first equations of the system to substitute the derivatives of $u_i$ and $u_j$. After these substitutions, the terms involving the derivatives of $u$ cancel out and we arrive at
$$
\frac{\delta {\cal{L}}_{ij}}{\delta u} \,=\,n_i\, n_j\Big( \partial_u \nabla_i \nabla_j \,H_{ij} + (\partial_u \partial_{u_i} \ln h_i) \nabla_j \,H_{ij} + (\partial_u \partial_{u_j} \ln h_j) \nabla_i \,H_{ij} \Big)\,.
$$
Using the third identity of Proposition \ref{identities-nabla}, one easily concludes that $\delta_u {\cal{L}}_{ij}=0$.

Denoting by $\cal{R}$ the left hand side of (\ref{closure-cont}), i.e.
$${\cal{R}}\,:=\, {\rm{D}}_{\alpha_i} {\cal{L}}_{j \ell}\,+\,{\rm{D}}_{\alpha_j} {\cal{L}}_{\ell i}\,+\,{\rm{D}}_{\alpha_\ell} {\cal{L}}_{ij}\,,$$
we will prove that ${\cal{R}}=0$ on solutions of $\Sigma_{ij}$, $\Sigma_{j\ell}$ and $\Sigma_{\ell i}$.

First, we work out explicitly $\cal{R}$ to show that the second order derivatives of $u_i$, $u_j$ and $u_\ell$ cancel out. Then, we factorize the remaining terms with respect to the derivatives of $u$, which leads to the following expression
\begin{eqnarray}
{\cal{R}} &=& \frac{\partial^2 u}{\partial \alpha_i \partial \alpha_j} \left( \frac{1}{h_j} \frac{\partial u_j}{\partial \alpha_\ell}-\frac{1}{h_i} \frac{\partial u_i}{\partial \alpha_\ell} + \Lambda_{\ell i} \frac{\partial u}{\partial \alpha_\ell}+ \Lambda_{j \ell} \frac{\partial u}{\partial \alpha_\ell} - \Xi_{i\ell} + \Xi_{j \ell}\right)\nonumber \\
&& +\,\,\,{\mbox{similar terms obtained by cycling permuting indices }}\,(i,\,j,\,\ell) \nonumber \\
&& + \,\,\,{\mbox{terms involving only first order derivatives of}}\,\,u\,, \label{cont-clos-lhs}
\end{eqnarray}
where
$$\Lambda_{ij} \,:=\, -\,\frac{H_{ij}}{h_i h_j}\,,\quad \Xi_{ij}\,:=\, -\,n_j \left(\frac{1}{h_i} \nabla_j \,H_{ij} +  \partial_u \ln h_j \right)\,.$$

The coefficient of $\partial_{\alpha_i} \partial_{\alpha_j} u$ in (\ref{cont-clos-lhs}) vanishes by taking into account the first order equations of $\Sigma_{j\ell}$ and $\Sigma_{\ell i}$. Due to the symmetry of $\cal{R}$, the coefficients of the remaining second order derivatives of $u$ also vanish.

On the other hand, the terms in (\ref{cont-clos-lhs}) involving only first order derivatives of $u$ take the following form:
\begin{eqnarray*}
&&\frac{1}{h_i h_j h_\ell} \frac{\partial u}{\partial \alpha_i} \, \frac{\partial u}{\partial \alpha_j}\, \frac{\partial u}{\partial \alpha_\ell} T_{ij\ell}  +  n_i n_j n_\ell\,\nabla_i\,\nabla_j\,\nabla_\ell\,T_{ij\ell} \\ 
&& + \left(\frac{n_\ell}{h_i h_j}\,\frac{\partial u}{\partial \alpha_i}\frac{\partial u}{\partial \alpha_j} \nabla_\ell + \frac{n_i}{h_j h_\ell}\,\frac{\partial u}{\partial \alpha_j}\frac{\partial u}{\partial \alpha_\ell} \nabla_i +  
\frac{n_j}{h_i h_\ell}\,\frac{\partial u}{\partial \alpha_\ell}\frac{\partial u}{\partial \alpha_i} \nabla_j \right) T_{ij \ell}\\
&& + \left(\frac{n_j n_\ell}{h_i}\,\frac{\partial u}{\partial \alpha_i} \nabla_j \nabla_\ell + 
\frac{n_\ell n_i}{h_j}\,\frac{\partial u}{\partial \alpha_j} \nabla_\ell \nabla_i +
\frac{n_i n_j}{h_\ell}\,\frac{\partial u}{\partial \alpha_\ell} \nabla_i \nabla_j \right)\, T_{ij\ell}.
\end{eqnarray*}
The multidimensional consistency of the continuous systems, i.e. identity (\ref{relG}), implies that the above expression vanishes on the solutions of systems $\Sigma_{ij}$, $\Sigma_{j\ell}$ and $\Sigma_{\ell i}$ and, subsequently, the closure relation (\ref{closure-cont}) is satisfied.
\end{proof}

\begin{acknowledgements}
P.X. is supported by the {\emph{Newton International Fellowship}} grant NF082473 entitled
``Symmetries and integrability of lattice equations and related partial differential equations'', which is run by The British Academy, The Royal Academy of Engineering and The Royal Society. S.L. was supported by the UK Engineering and Physical Sciences Research Council (EPSRC). 
\end{acknowledgements}

\appendix{The list of the ABS equations and their characteristic polynomials} \label{fGk}

The following constitutes the list of scalar affine-linear multidimensionally consistent quadrilateral lattice equations up to M\"obius transformations that was established by Adler, Bobenko \& Suris (2003). The equations of the form A1, A2 of their paper are omitted as they are related to other members of the list by point transformations. The form of Q4 given below was established by Hietarinta (2005), in which ${\rm{sn}}$ denotes the Jacobi elliptic function ${\rm{sn}}(x|k)$ with modulus $k$.

\begin{tabular}{lr} 
& \\
$(u -u_{ij})\, (u_i-u_j)\, -\,\alpha_i \,+ \, \alpha_j \, = \,0$ & (H1) \\
&  \\
$(u -u_{ij})(u_i-u_j) +(\alpha_j-\alpha_i) (u +u_i+u_j+u_{ij})- \alpha_i^2 + \alpha_j^2 = 0$ & (H2) \\
& \\
${\rm{e}}^{-\alpha_i/2} (u u_i+u_j u_{ij}) - {\rm{e}}^{-\alpha_j/2} (u u_j+u_i u_{ij}) + \delta ({\rm{e}}^{-\alpha_i}-{\rm{e}}^{-\alpha_j}) = 0 $ & (H3)\\
& \\
\end{tabular} 

\begin{tabular}{lr} 
$\alpha_i (u -u_j) (u_i- u_{ij}) - \alpha_j (u - u_i) (u_j -u_{ij}) + \delta^2 \alpha_i \alpha_j (\alpha_i-\alpha_j)= 0$ & (Q1) \\
& \\
$\alpha_i (u -u_j) (u_i- u_{ij}) - \alpha_j (u - u_i) (u_j -u_{ij}) +$ & \\
$ \alpha_i \alpha_j (\alpha_i-\alpha_j) (u +u_i+u_j+u_{ij}) - \alpha_i \alpha_j (\alpha_i-\alpha_j) (\alpha_i^2-\alpha_i \alpha_j + \alpha_j^2) = 0$ & (Q2)\\
& \\
$({\rm{e}}^{-\alpha_j}-{\rm{e}}^{-\alpha_i}) (u u_{ij}+u_i u_j) + {\rm{e}}^{-\alpha_j/2} ({\rm{e}}^{-\alpha_i}-1) (u u_i+u_j u_{ij})$ & \\
$- {\rm{e}}^{-\alpha_i} ({\rm{e}}^{-\alpha_j}-1) (u u_j+u_i u_{ij}) - \frac{\delta^2 ({\rm{e}}^{-\alpha_i}-{\rm{e}}^{-\alpha_j}) ({\rm{e}}^{-\alpha_i}-1) ({\rm{e}}^{-\alpha_j}-1)}{4 {\rm{e}}^{-(\alpha_i +\alpha_j)/2}}=0 $ & (Q3) \\
& \\
$ {\rm{sn}}(\alpha_i) (u u_i+u_j u_{ij}) - {\rm{sn}}(\alpha_j) (u u_j+u_i u_{i j}) -{\rm{sn}}(\alpha_i-\alpha_j) (u u_{ij}+u_i u_j) +$ & \\
$ k\,{\rm{sn}}(\alpha_i)\, {\rm{sn}}(\alpha_j) \,{\rm{sn}}(\alpha_i-\alpha_j) \,(1 + u u_i u_j u_{ij}) = 0$ & (Q4)\\
& \\
\end{tabular}

The formulae for the canonical bi-quadratics $h$ and $H$, and the corresponding discriminant curves, for the ABS equations are collected in the following lists.
\begin{flushleft}
\begin{tabular}{llll}
{\it{H1}} : & $h_i = 1$ & $\kappa_{ij} = \alpha_j-\alpha_i$ & $H = (u_i-u_j)^2$  \\
{\it{H2}} : & $h_i = 2(u + u_i + \alpha_i)$  & $\kappa_{ij} = \alpha_j-\alpha_i$& $H = (u_i-u_j)^2 - (\alpha_i-\alpha_j)^2$  \\
{\it{H3}} : & $h_i = u u_i + {\rm{e}}^{-\alpha_i/2} \delta$  & $\kappa_{ij} = {\rm{e}}^{-\alpha_i}- {\rm{e}}^{-\alpha_j}$ & $H= {\rm{e}}^{-(\alpha_i+\alpha_j)/2} (u_i^2+u_j^2) $\\
& & & $\qquad -({\rm{e}}^{-\alpha_i}+{\rm{e}}^{-\alpha_j}) u_i u_j$ \\
& & & 
\end{tabular}

\begin{tabular}{cl}
{\it{Q1}} : &  $ h_i \,= \,((u-u_i)^2 - \alpha_i^2 \delta^2)/\alpha_i \,,\quad \kappa_{ij} \,=\, \alpha_i \alpha_j (\alpha_j-\alpha_i)$ \\
& \\
& $H = \alpha_i \alpha_j \left((u_i-u_j)^2 - (\alpha_i -\alpha_j)^2 \delta^2\right)$ \\
& \\
\end{tabular}

\begin{tabular}{cl}
{\it{Q2}} : & $ h_i = ((u-u_i)^2 - 2 \alpha_i^2 (u+u_i) + \alpha_i^4)/\alpha_i\,,\quad \kappa_{ij} = \alpha_i \alpha_j (\alpha_j-\alpha_i)$ \\
& \\
            & $ H = \alpha_i \alpha_j \left((u_i-u_j)^2 - 2 (\alpha_i -\alpha_j)^2 (u_i+u_j) + (\alpha_i-\alpha_j)^4\right)$\\
            & \\
\end{tabular}

\begin{tabular}{cl}
{\it{Q3}} : &  $h_i =\frac{\exp(-\alpha_i/2)}{1-\exp(-\alpha_i)}(u^2+u_i^2) - \frac{1+\exp(-\alpha_i)}{1-\exp(-\alpha_i)} u u_i  + \delta^2 \frac{1-\exp(-\alpha_i)}{4 \exp(-\alpha_i/2)}$ \\
    & \\        
            & $\kappa_{ij}\, =  \,\left[\exp(-\alpha_i) - 1\right]\, \left[\exp(-\alpha_j)-1\right]\,\left[\exp(-\alpha_i) - \exp(-\alpha_j)\right]$ \\
        & \\    
            & $H =({\rm{e}}^{-\alpha_i}-1)({\rm{e}}^{-\alpha_j}-1)\Big({\rm{e}}^{-(\alpha_i+\alpha_j)/2} (u_i^2+u_j^2) - ({\rm{e}}^{-\alpha_i}+{\rm{e}}^{-\alpha_j}) u_i u_j $ \\
            & $ {\phantom{G_{ij} =({\rm{e}}^{-\alpha_i}-1)\,({\rm{e}}^{-\alpha_j}-1)}} + \frac{1}{4}\, \delta^2 \, {\rm{e}}^{(\alpha_i+\alpha_j)/2}\, ({\rm{e}}^{-\alpha_i}-{\rm{e}}^{-\alpha_j})^2\, \Big)$\\
            & \\
\end{tabular}

\begin{tabular}{cl}
{\it{Q4}} : &  $h_i = -k\, {\rm{sn}}(\alpha_i) (1+u^2 u_i^2) + \frac{1}{{\rm{sn}}(\alpha_i)} \left( u^2 +u_i^2 - 2 {\rm{cn}}(\alpha_i) {\rm{dn}}(\alpha_i) u u_i \right)$\\
            & \\       
                   & $ \kappa_{ij} = \,{\rm{sn}}(\alpha_i){\rm{sn}}(\alpha_j){\rm{sn}}(\alpha_j-\alpha_i)$ \\
                & \\
                   & $H = {\rm{sn}}(\alpha_i){\rm{sn}}(\alpha_j) \left\{-k\, {\rm{sn}}^2\,(\alpha_i-\alpha_j) (1 +u_i^2 u_j^2) + u_i^2 + u_j^2\right\} + $ \\
                   & $\left({\rm{sn}}^2(\alpha_i-\alpha_j)- {\rm{sn}}^2(\alpha_i) - {\rm{sn}}^2(\alpha_j) + k^2 {\rm{sn}}^2(\alpha_i-\alpha_j) {\rm{sn}}^2(\alpha_i) {\rm{sn}}^2(\alpha_j) \right) u_i u_j$ \\
                   & \\
\end{tabular}
\end{flushleft}

\begin{tabular}{c|ccccccc}
 & {\it{H1}} & {\it{H2}} & {\it{H3}} & {\it{Q1}} & {\it{Q2}} & {\it{Q3}} & {\it{Q4}} \\ \hline
$r(u)$ & 0 & 4 & $u^2$ & $4 \delta^2$ & $16 u$ & $u^2 - \delta^2$ & $4 \left(k u^4-(1+k^2) u^2 + k \right)$ 
\end{tabular}

\appendix{Proof of Lemma \ref{f-omega-relations}}%the fundamental relation of polynomials $h_i$}
\label{aDerivative}

As was proven in Tongas {\it{et al.}} (2007), all of the ABS equations admit three point generalized symmetries which have the following form
$$ K^{(1)}_i\,=\,P\big(h_i(u,u_i)\big)\,\partial_u \,:=\, \left\{\left( \frac{1}{u_i-u_{-i}}- \frac{1}{2} \partial_{u_i}\right) h_i(u,u_i)\right\}\,\partial_u\,.$$
Moreover, it was shown that any other three-point generalized symmetry $K$ is necessarily of the form $ K\,=\,\left\{a(n)\,K^{(1)}_i\,+\,\phi(n,m,u)\right\} \partial_u\,,$
where the functions $a(n)$, $\phi(n,m,u)$ are determined by the solutions of a linear equation. Additionally, it was proven that all of the ABS equations admit extended generalized symmetries, namely $V_i$ given in (\ref{master-sym}).

On the other hand, it was shown in Mikhailov {\it{et al.}} (2010) that a higher order generalized symmetry is given by
$$ K^{(2)}_i \,=\, \left(\frac{h_i(u,u_i) h_i(u,u_{-i})}{(u_i-u_{-i})^2}\,\left({\cal{S}}_i + {\cal{S}}_{i}^{-1} \right) \frac{1}{u_i-u_{-i}}\right) \,\partial_u\,.$$

The commutator of the extended symmetry $V_i$ with the first generalized symmetry $K_i^{(1)}$ yields
$$[{\bf{v}}_i,K^{(1)}_i]\,=\,-K^{(2)}_i\,+\,P\left(\partial_{\alpha_i} h_i(u,u_i) + \frac{1}{2} \left|\begin{array}{cc} h_i & \partial_u h_i \\ \partial_{u_i} h_i & \partial_u \partial_{u_i} h_i \end{array} \right|\right)\,,$$
which is another symmetry of the equation under consideration. Since $K^{(2)}$ is a symmetry of the equation, the same must hold for the remaining part of the above commutator. This means that the latter has to be of the general form $K$ which implies
$$\partial_{\alpha_i} h_i(u,u_i) + \frac{1}{2} \left|\begin{array}{cc} h_i & \partial_u h_i \\ \partial_{u_i} h_i & \partial_u \partial_{u_i} h_i \end{array} \right|\,=\,\mu(\alpha_i)\,h_i(u,u_i)\,.$$
But, for all of the ABS equations, it turns out that $\mu \equiv 0$.

\label{lastpage}

\end{document}